\documentclass{article}
\usepackage{amsthm}
\usepackage{amsmath}
\usepackage{amssymb}
\usepackage{tikz}
\usetikzlibrary{automata, matrix, positioning, fit}

\title{Product gales and Finite state dimension}
\author{ Akhil S }
\date{ }

\theoremstyle{definition}
\newtheorem{definition}{Definition}
\newtheorem{lemma}{Lemma}
\newtheorem{theorem}{Theorem}
\newcommand{\dimfs}{\mathrm{dim}_\mathrm{FS}}
\newcommand{\N}{\mathbb{N}}
\newcommand{\Q}{\mathbb{Q}}
\newcommand{\R}{\mathbb{R}}
\newcommand{\F}{\mathcal{F}}
\newcommand{\prefix}{\sqsubseteq}
\newcommand{\mbet}{\mathrm{multi-bet}}

\begin{document}

\maketitle

\begin{abstract}

In this work, we introduce the notion of product gales, which is the modification of an $s$-gale such that $k$ separate bets can be placed at each symbol. The product of the bets placed are taken into the capital function of the product-gale. We show that Hausdorff dimension can be characterised using product gales. 

A $k$-bet finite-state gambler is one that can place $k$ separate bets at each symbol. 
We call the notion of finite-state dimension, characterized by product gales induced by $k$-bet finite-state gamblers, as multi-bet finite-state dimension. Bourke, Hitchcock and Vinodchandran \cite{bourke2005entropy}  gave an equivalent characterisation of finite state dimension by disjoint block entropy rates. 
We show that multi-bet finite state dimension can be characterised using  sliding block entropy rates. Further, we show that multi-bet finite state dimension can also be charatcterised by disjoint block entropy rates. 

Hence we show that finite state dimension and multi-bet finite state dimension are the same notions, thereby giving a new characterisation of finite state dimension using $k$-bet finite state $s$-gales. 
We also provide a proof of equivalence between sliding and disjoint block entropy rates, providing an alternate, automata based proof of  the result by Kozachinskiy, and Shen \cite{KozachinskiyShen2021}.

\end{abstract}

\section{Introduction}

The notion of Hausdorff Dimension introduced by Felix Hausdorff \cite{Hausdorff1919} is a powerful tool in fractal geometry and geometric measure theory. Hausdorff introduced the notion of $s$-dimensional outer measures by generalising lebesgue outer measure  for any value of $s \in [0, \infty)$. Using this notion, Hausdorff generalised the framework of dimension to non integer values. 

Lutz \cite{Lutz03} introduced the notion of $s$-gales, which are betting strategies that generalise martingales and gave a characterisation of Hausdorff dimension using $s$-gales. 
Lutz \cite{Lutz03} further
effectivised the notion of $s$-gales and introduced the notion of constructive Hausdorff dimension. Dai, Lathrop, Lutz, and
Mayordomo \cite{Dai2001} carried forward this effectivisation to the level of finite-state automata.

Dai, Lathrop, Lutz, and
Mayordomo \cite{Dai2001} introduced the notion of finite state dimension using the finite state analogues of $s$-gales. They show that finite state dimension is equivalent to compression rates by finite state automata. Hence they establish that finite state dimension quantifies the information density of a sequence as measured by finite state computation \cite{Dai2001}.  Finite state dimension of a sequence is a real number that lies between $0$ and $1$. Sequences having finite state dimension $0$ are those that are predictable by finite state automata.
Whereas, those with finite state dimension 1 are random to finite state automata.
The concept of finite state dimension also has connections to metric number theory.  A  sequence $X$ drawn from an alphabet $\Sigma$ is said to be normal if every block of length $\ell$ occurs with asymptotic probability $|\Sigma|^{-\ell}$. Due to the result by Schnorr and Stimm \cite{SchSti72}, a sequence $X \in \Sigma^\infty$ is normal if and only if finite state dimension of $X$ is $1$.

Bourke, Hitchcock and Vinodchandran \cite{bourke2005entropy} further explored this connection between finite state dimension of a sequence and the distribution of the substrings within it.
For a fixed length $\ell \in \N$, the limiting (limit infimum) value of Shannon entropy of the distribution of aligned (non intersecting or disjoint) $\ell$-length blocks, over the finite prefixes of the sequence is considered. This value is normalised by dividing by $\ell$ and the infimum of this value taken over all lengths $\ell$ is called the disjoint block entropy rate of the sequence.
Bourke, Hitchcock and Vinodchandran \cite{bourke2005entropy} showed the equivalence between finite state dimension and disjoint block entropy rates, hence bridging the results by Ziv and Lempel \cite{ZivLem78} and Dai, Lathrop, Lutz, and
Mayordomo \cite{Dai2001}.

S. S. Pillai \cite{Pillai1940} showed that the notion of normality remains unchanged weather the occurence of any blocks $w \in \Sigma^\ell$ are counted in an aligned (disjoint) fashion, or in a non-aligned (sliding, overlapping) fashion.
Kozachinskiy and Shen \cite{KozachinskiyShen2021} generalised the result by Pillai showing that the notions of sliding and disjoint block entropy rates are also equivalent, using information theoretic techniques.

In this work, we introduce the notion of product gales. A $k$-product $s$-gale is a product of $k$ different $s$-gales. It is fundamentally different from the notion of $k$-account $s$-gales \cite{Dai2001} as the product of the $k$-accounts are taken instead of the sum.
Therefore, a product gale need not necessarily be effectively an $s$-gale. We show that Hausdorff dimension can be characterised using product-gales, generalising the result by Lutz \cite{Lutz03}. 

We then introduce the notion of $k$-bet finite state gamblers, in which we allow the gambler to place $k$ separate bets on each symbol in the sequence. The resultant $k$-product $s$-gale is obtained by multiplying together the $s$-gales induced by the $k$ bets. We define a notion of multi-bet finite state dimension using $k$-bet $s$-gales. We show that this notion perfectly captures the process of betting on a sequence in sliding block fashion. 
We show that multi-bet finite state dimension is equivalent to the sliding block entropy rates. 

We then show that  multi-bet finite state dimension is equivalent to disjoint block entropy rates.
Hence we obtain that the notion of  multi-bet finite state dimension is equivalent to ordinary finite state dimension. Thereby, we obtain a new characterisation of finite state dimension. We also obtain that the notions of sliding and disjoint block entropy rates are equivalent, thereby giving a new automata based proof of Kozachinskiy and Shen \cite{KozachinskiyShen2021}.

\section{Preliminaries}

\subsection{Notation}
$\Sigma$ denotes a finite alphabet. $\Sigma^ \ell$ denotes strings of size $\ell$ using $\Sigma$.  $\Sigma^\infty$ denotes infinite strings and $\Sigma^*$ denotes the set of finite strings formed using $\Sigma$.   
$\Sigma^{<\ell}$ denotes all strings of $\Sigma$ of size less than $\ell$, including the empty string $\lambda$. 

Given $X \in \Sigma^\infty$, $X \restriction n$ denotes the string formed by the first $n$ symbols of $X$. For a finite string $x \in \Sigma^*$, $x[a:b]$ denotes the string $x$ from index $a$ to $b-1$.
We use the same notion for finite strings as well. $\log$ denotes the logarithm to the base 2.

Given $w \in \Sigma^*$, the cylinder set of $w$, denoted $C_w$ is the set of strings in $\Sigma^\infty$ for which $w$ appears as a prefix.

$\N$ denotes the set of natural numbers, $\Q$ denotes the set of rational numbers and $\R$ denotes the set of real numbers. $\prefix$ denotes the prefix of operation, which includes equality.

\subsection{Hausdorff Dimension over Cantor Space} \label{sec:HausdorffDimension}

The following definitions are originally given by Hausdorff \cite{Hausdorff1919}. We take the adaptation of the following for the Cantor space given by Lutz \cite{Lutz03}. 

\begin{definition}  [Hausdorff \cite{Hausdorff1919}, Lutz \cite{Lutz03}]
	Given a set $\F \subseteq \Sigma^\infty$, a 
	collection of strings $\{w_i\}_{i \in \N}$ where for each $i \in \N$, $w_i \in \Sigma^*$  is
	called a \emph{$\delta$-cover of $\F$} if for all
	$i \in \N$, $2^{-|w_i|} \leq \delta$ and $\F \subseteq
	\bigcup_{i\in\N}{C_{w_i}}$.
\end{definition}

\begin{definition}  [Hausdorff \cite{Hausdorff1919}, Lutz \cite{Lutz03}]
	Given an $\F \subseteq \Sigma^\infty$, for any $s>0$, define
	$$
	\mathcal{H}^s_\delta(\F) = \inf\left\{\sum_i 2^{-s.|w_i|} :
	\{w_i\}_{i\in\N} \text{ is a } \delta\text{-cover of } \F
	\right\}.
	$$
\end{definition}

As $\delta$ decreases, the set of admissible $\delta$ covers decreases.
Hence $\mathcal{H}^s_\delta(\F)$ increases.

\begin{definition}  [Hausdorff \cite{Hausdorff1919}, Lutz \cite{Lutz03}]
	For $s \in (0,\infty)$, the \emph{s-dimensional Hausdorff outer 
		measure} of $\F$ is defined as:
	$$\mathcal{H}^s(\F) = \lim\limits_{\delta \rightarrow 0}
	\; \mathcal{H}^s_\delta(\F).$$
\end{definition}


Observe that for
any $t>s$, if $\mathcal{H}^s(\F) < \infty$, then
$\mathcal{H}^t(\F) = 0$ (see Section 2.2 in \cite{Falc03}).

Finally, we have the following definition of Hausdorff dimension over $\Sigma^\infty$.

\begin{definition} [Hausdorff \cite{Hausdorff1919}, Lutz \cite{Lutz03}]
	For any $\F \subset \Sigma^\infty$, the \emph{Hausdorff dimension} of $\F$ is defined as:
	$$\dim(\F) = \inf \{s\geq 0 : \mathcal{H}^s(\F) = 0 \}
	.$$
\end{definition}

\subsection{$s$-gales}

\begin{definition} [$s$-gale \cite{Lutz03}] \label{def:s-gale}
	An s-gale is a function $d :\{0,1\}^* \rightarrow [0,\infty)$, that satisfies the following condition:
	\[d(w) = 2^{-s} [d(w0) + d(w1)]\]
	for all $w \in \{0,1\}^*$.
\end{definition}

\begin{definition} [Success set of $s$-gale \cite{Lutz03}]
	Let $d$ be an s-gale.
	\begin{enumerate}
		\item $d$ is said to succeed on a sequence $X \in \Sigma^\infty$ if:
		
		\[\limsup_{n \rightarrow \infty} d(X \restriction n ) = \infty.\]
		\item The success set of $d$ is
		
		\[S^\infty[d] = \{X \in \Sigma^\infty : d \text{ succeeds on } X\}.\]
	\end{enumerate}
\end{definition}

\subsection{Finite-state Dimension}

\begin{definition}[Finite-State Gambler \cite{Dai2001}]
	A finite-state gambler (FSG) is a 5-tuple
	\[
	G = (Q, \delta, \beta, q_0, c_0),
	\]
	where
	\begin{itemize}
		\item \( Q \) is a nonempty, finite set of states,
		\item \( \delta : Q \times \{0, 1\} \to Q \) is the transition function,
		\item \( \beta : Q \to \Q \cap [0,1] \) is the betting function,
		\item \( q_0 \in Q \) is the initial state, and
		\item \(c_0\), is the initial capital, a nonnegative rational number.
	\end{itemize}
\end{definition}

\begin{definition} [$s$-gale induced by FSG \cite{Dai2001}]
Given $s \in [0,\infty)$, the $s$-gale induced by the FSG \( G \) is the function 
	\[
	d_{G}^s : \{0, 1\}^* \to [0, \infty)
	\]
	defined by the recursion:
	\[
	d_{G}^s(\lambda) = c_0,
	\]
	and
	\[
		d_{G}^s(wb) = 2^s \cdot 	d_{G}^s (w) \cdot \left[(1 - b) \cdot (1 - \beta(w)) + b\cdot\beta(w)\right].
	\]
\end{definition}

\begin{definition}
	Let $\mathcal{G}_{FS}(X)$ denote the set of all $s \in[0,1]$ such that there is a finite state s-gale $d^s_G$ for which $X \in S^\infty[d_G^s]$.
\end{definition}

\begin{definition} [Finite state dimension \cite{Dai2001}]\label{def:FSD}
	 The Finite state dimension of a sequence $X \in \Sigma^\infty$ is
	\[\dimfs(X) = \inf \mathcal{G}_{FS}(X).\]
\end{definition}

\subsection{Disjoint Block Entropy rates}

\begin{definition}
	Given $x \in
	\Sigma^{\ell\cdot k}$  and $u 
	\in \Sigma^\ell$, let
	\begin{align*}
		N(u,x) = \lvert\{0\leq i < k \; : \;
		x[i\ell : (i+1)\ell] = u \}\lvert.
	\end{align*}
\end{definition}

Therefore, $N(u,x)$ is the number of times $u$ occurs in $\ell$ length blocks of
$x$ in a sliding manner.

\begin{definition}
	The \emph{disjoint block frequency} of $u 
	\in \Sigma^\ell$
	in  $x \in \Sigma^{\ell\cdot k}$  and  is defined as :
	\begin{align*}
		P(u , x) = \frac{N(u,x)}{k}.
	\end{align*}
\end{definition}

It follows
that $\sum\limits_{u \in \Sigma^\ell} P(u ,x)= 1$.

\begin{definition}
	For a finite sequence $x \in \Sigma^{\ell \cdot k}$ , the \emph{$\ell$-length disjoint block entropy rate} of $x$ is defined as:
	\begin{align*}
		H_\ell^{disjoint}(x) = - \frac{1}{\ell \; \log(|\Sigma|)}
		 \sum_{u \in
			\Sigma^\ell} P(u ,x ) \; \log(P(u ,x )).
	\end{align*}
	
\end{definition}

\begin{definition} [Disjoint block Entropy rate \cite{bourke2005entropy}]
	For $X \in \Sigma^\infty $, the
	$\ell$-length \emph{disjoint block entropy rate} of $X$ is
	defined as:
	\begin{align*}
		H_\ell^{disjoint}(X) = 
		\liminf\limits_{k\rightarrow \infty} H_\ell^{disjoint}(X \restriction k\ell).
	\end{align*}	
	The \emph{ disjoint block entropy rate} of $X$ is defined
	as\footnote[1]{It holds that the $\inf\limits_{\ell}$ can be
		replaced with a stronger $\lim\limits_{\ell \rightarrow \infty}$ for
		  block entropy rate \cite{KozachinskiyShen2019}.} :
	\begin{align*}
		H^{disjoint}(X) = \inf\limits_{\ell \in \N} H_\ell^{disjoint}(X).
	\end{align*}	 
\end{definition}

\subsection{Sliding block entropy rates}

Sliding block entropy rate is an analogous notion of disjoint block entropy rates. The key difference is that the counting of blocks is done is a sliding (non aligned) manner. 

\begin{definition}
Given $x \in
\Sigma^{n}$  and $u 
\in \Sigma^\ell$,  such that $n \geq \ell$, let
\begin{align*}
	N(u,x) = \lvert\{0\leq i < n - \ell \; : x[i : i + \ell] 
	= u \}\lvert.
\end{align*}
\end{definition}

Therefore, $N(u,x)$ is the number of times $u$ occurs in $\ell$ length blocks of
$x$ in a sliding manner.

\begin{definition}
	The \emph{sliding block frequency} of $u 
	\in \Sigma^\ell$
	in  $x \in \Sigma^{n}$  and  is defined as :
	\begin{align*}
		P(u , x) = \frac{N(u,x)}{n - \ell + 1}.
	\end{align*}
\end{definition}

It follows
that $\sum\limits_{u \in \Sigma^\ell} P(u ,x)= 1$.

\medskip

\begin{definition}
	For a finite sequence $x \in \Sigma^{n}$ , the \emph{$\ell$-length sliding block entropy rate} of $x$ is defined as:
	\begin{align*}
		H_\ell^{sliding}(x) = - \frac{1}{\ell \; \log(|\Sigma|)}
		\sum_{u \in
			\Sigma^\ell} P(u ,x ) \; \log(P(u ,x )).
	\end{align*}
	
\end{definition}

\begin{definition} [Sliding block entropy rates \cite{KozachinskiyShen2019}]
	For $X \in \Sigma^\infty $, the
	$\ell$-length \emph{sliding block entropy rate} of $X$ is
	defined as:
	\begin{align*}
		H_\ell^{sliding}(X) = 
		\liminf\limits_{n\rightarrow \infty} H_\ell^{sliding}(X \restriction n).
	\end{align*}	
	The \emph{sliding block entropy rate} of $X$  is defined
	as\footnote[1]{It holds that the $\inf\limits_{\ell}$ can be
		replaced with a stronger $\lim\limits_{\ell \rightarrow \infty}$ for
		block entropy rate\cite{KozachinskiyShen2019}.}  :
	\begin{align*}
		H^{sliding}(X) = \inf\limits_{\ell \in \N} H_\ell^{sliding}(X).
	\end{align*}	 
\end{definition}

\subsection{Block Entropy rates and  Finite state Dimension}
Bourke, Hitchcock and Vinodchandran \cite{bourke2005entropy} show the following disjoint block entropy characterisation of finite state dimension. We give an alternate automata based proof. It is the simplification of the proof given in Akhil, Nandakumar and Pulari \cite{fsrdconf}. 

First we prove a lemma that gives a probability distribution $\mathbb{P}$ from the block entropy rate, which we use in the proof of Theorem \ref{thm:DisjEquiv}. 

\begin{lemma} 
	\label{lem3.1} 
	For all $X \in \Sigma^\infty$ and $ s' > H^{disjoint}(X)$, there exists a block length  $\ell \in \mathbb{N}$, and a probability distribution function $\mathbb{P} : \{0,1\}^\ell \rightarrow \Q \cap [0,1]$, such that for infinitely many $k \in \mathbb{N}$,
	
	
	\[  - \sum_{w \in \{0,1\}^\ell}  P(w,X \restriction k\ell) \; \log(\mathbb{P}(w))) < \ell \; (s' - d)\]
	
	for some $d < s'$.
\end{lemma}

\begin{proof}
	Since, $\inf\limits_{l} H_l(X) = s$, for any $ s'>s, \exists \; \ell \; \text{ s.t }  H_\ell(X) < s'$.
	
	At that $\ell$, let $s' - H_\ell(X) = 3d$. So, $ \liminf\limits_{k\rightarrow \infty} H_\ell(X|k\ell) = s' - 3d$.
	
	\medskip
	
	There exists an infinite number of $k \in \N$ such that $H_\ell(X|k\ell) < s' - 2d$. 
	
	For each such $k$ we have,  
	\begin{align} \label{eq:1}
		- \frac{1}{\ell}  \sum_{w \in \{0,1\}^\ell}  P(w,X \restriction k\ell) \; \log(P(w,X\restriction k\ell)) < s'- 2d.
	\end{align} 
	
	Since $0 \leq  P(w,X \restriction k\ell) \leq 1$ is bounded, there exists a convergent subsequence, corresponding to indices $k' \in \mathbb{N}$ and let $\{\mathbb{P}(w)\}_{w \in \Sigma^\ell} $ be the limit of the subsequence. 
	
	From (\ref{eq:1}) and continuity of entropy function, it follows that,
	
	\begin{align} \label{eq:2}
		- \frac{1}{\ell} \sum_{w \in \{0,1\}^\ell} \mathbb{P}(w) \log(\mathbb{P}(w)) < s'-2d.
	\end{align}
	
	For all $\epsilon > 0$, there exists a $k_\epsilon$ such that for all such $k' > k_\epsilon$,
	\[P(w,X\restriction k'\ell) < \mathbb{P}(w) + \epsilon. \]
	
	So for infinitely many $k' \in \N$ ,
	
	\begin{align*}
		- \frac{1}{\ell} \sum_{w \in \{0,1\}^\ell}   P(w,X\restriction k'\ell) \;& \log(\mathbb{P}(w)) < 
		\frac{1}{\ell}  \sum_{w \in \{0,1\}^\ell}  (\mathbb{P}(w) + \epsilon) \; \log(1/\mathbb{P}(w))\\
		&= - \frac{1}{\ell} \sum_{w \in \{0,1\}^\ell}  \mathbb{P}(w)  \log(\mathbb{P}(w)) + \frac{\epsilon}{\ell}  \sum_{w \in \{0,1\}^\ell} \log(1/\mathbb{P}(w)).
	\end{align*}
	
	Applying (\ref{eq:2}), for infinitely many $k' \in \N$, \[ - \frac{1}{\ell}. \sum_{w \in \{0,1\}^\ell}  P(w,X \restriction k'\ell) \; \log(\mathbb{P}(w)) < s' -2d + \frac{\epsilon}{\ell}  \sum_{w \in \{0,1\}^\ell} \log(1/\mathbb{P}(w)).\]
	
	Given $\epsilon' > 0$, let $\mathbb{P'}(w): \Sigma^\ell \to \Q \cap [0,1]$ be a probability distribution such that forall $w \in \Sigma^\ell$, $|\log(\mathbb{P'}(w)) - \log(\mathbb{P}(w))| < \epsilon'$.
	
	Taking small enough $\epsilon$, $\epsilon'$, for infinitely many $k' \in \N$,
	
	\[- \frac{1}{\ell}. \sum_{w \in \{0,1\}^\ell}  P(w,X \restriction k'\ell) \; \log(\mathbb{P'}(w)) <  s' - d.\]
\end{proof}

For any $ s > H^{disjoint}(X)$, we use $\mathbb{P}(w)$'s obtained from Lemma \ref{lem3.1} to construct an $s-gale$ that bets on strings of length $\ell$ and succeeds on $X$.

\begin{lemma} 
	\label{lem3.3} 
	
	For all $X \in \Sigma^\infty$,
	
	\[\dimfs(X) \leq H^{disjoint}(X) \]
\end{lemma}

\begin{proof}
	For any $s >H^{disjoint}(X)$, consider the probability distribution distribution $\mathbb{P}(w) : \Sigma^\ell \to [0,1] \cap \Q$ obtained from Lemma \ref{lem3.1}.
	
	Consider the following finite state gambler,
	$G_{\mathbb{P}} = (Q,\delta,\beta,q_0)$, whose components are :
	
	\begin{itemize}
		\item $Q = \Sigma^{< \ell}$.
		\item For all $w \in Q$ and $b \in \Sigma,$
		
		\[
		\delta(w,b) = 
		\begin{cases}
			wb,& \text {if } |w| <  \ell-1\\
			\lambda, & \text {if } |w| =  \ell-1\\
		\end{cases}
		\]
		
		\item The betting is done as to ensure that the cumulative bet $\beta (w) = \mathbb{P}(w)$ for all $w \in \Sigma^\ell$.
		
		\begin{align*}
			\beta(w) &= \frac{\mathbb{P}(w.1)}{\mathbb{P}(w)}.
		\end{align*}
		
		where for $v \in \Sigma^{< \ell}$, $\mathbb{P}(v) = \sum\limits_{w \in \Sigma^\ell : v \prefix w} \mathbb{P}(w)$.
		
		\item $q_0 = \lambda$.
		
	\end{itemize}

	Consider the $s$-gale induced by $G_{\mathbb{P}}$. If disjoint block $w_i \in \Sigma^\ell$ has appeared $n_i$ times in $X \restriction k\ell$, then
	
	\[d_{G_{\mathbb{P}}}^{s}(X \restriction k\ell) = \prod_{w_i \in  \Sigma^\ell} (2^{s.\ell} * \mathbb{P}(w_i)) ^{n_i}\]
	
	Taking logarithm,
	
	\begin{align*}
		\log(d_{G_{\mathbb{P}}}^{s}(X \restriction k\ell))  &= \sum_{w_i \in  \{0,1\}^l} {n_i} \;(s\cdot \ell  + \log (\mathbb{P}(w_i)) ) \\
		&= k . \sum_{w_i \in  \{0,1\}^l} \frac{n_i}{k} \; (s\cdot \ell  + \log (\mathbb{P}(w_i)) )\\ 
		&= k . \sum_{w_i \in  \{0,1\}^l} P(w_i,X \restriction k\ell) \; (s\ell + \log(\mathbb{P}(w_i)))\\
		&=  k . \left(s\ell - \sum_{w_i \in  \{0,1\}^l} P(w_i,X\restriction k\ell) \; \log(1/{\mathbb{P}(w_i)})\right).
	\end{align*}
	
	From Lemma \ref{lem3.1}, we have for some $d < s$, for infinitely many $k \in \N$,
	
	\begin{align*}
		\log(d_{G_{\mathbb{P}}}^{s}(X \restriction k\ell)) \;\; 
		&\geq  k . (s\ell - \ell(s-d) )\\
		&=  k.d.\ell.
	\end{align*}
	
	Since $k$ is unbounded, we see that the $s$-gale induced by $G_{\mathbb{P}}$ succeeds on X.
	
	Hence, for any $s >H^{disjoint}(X)$, $\dimfs(X) \leq s$. 
	
	From this, it follows that for all $X \in \Sigma^\infty$, $H^{disjoint}(X) \geq \dimfs(X)$.
	
\end{proof}

We now bound the entropy rates of a sequence $X$ from the win of an $s$-gale on it. 

\begin{lemma}
	\label{lem3.6}
	
	For every $X \in \Sigma^\infty$,
	
	\[ H^{disjoint}(X)  \leq \dimfs(X). \]
\end{lemma}

\begin{proof}
	
	Let a finite state s-gale $d_{G}^{(s)}$ corresponding to a  gambler $G = (Q,\delta,\beta,q_0)$ succeed on a sequence $X \in \Sigma^\infty$. Given $L \in \N$, consider the extended $k$-bet finite state gambler,
	$G_L= (Q',\delta',\beta',q_0')$, whose components are
	
	\begin{itemize}
		\item $Q' = Q \times [L] $
		\item For all $(q,n) \in Q'$  and $b \in \Sigma$,
		
		\[
		\delta'((q,n),b) = 
		\begin{cases}
			(\delta(q,b),n+1)& \text {if } n < L-1\\
			(\delta(q,b),0) & \text {if } n = L-1\\
		\end{cases}
		\]
		
		\item For all $(q,n) \in Q'$,
		
		\[
		\beta((q,n)) = \beta(q).
		\]
		
		\item $q_0' = (q_0,0)$
		
	\end{itemize}
	
	It is easy to see that if the s-gale $d_{G}^{(s)}$ succeeds on $X$, then the $s$-gale $d_{G_L}^{(s)}$ succeeds on $X$ as the same bets are placed by $G$ and $G'$.
	
	Now for any $k \in \mathbb{N}$, 
	
	\begin{align*}
		d_{G_L}^{(s)}(X \restriction {kL}) =  \prod_{q \in Q} \prod_{x \in  \{0,1\}^L}  \{ {2^{s.L} \times \beta(q,x)} \} ^{N(q,x)}
	\end{align*}

	where $N(q,x)$ is the number of times $x$ has occurred as the input when $G_L$ was at state $(q,0)$. $\beta(q,x)$ is the cumulative bet placed by $G_L$ on $x$ in that scenario.
	
	Taking logarithm,
	
	\begin{align*}
		\log (d_{G_L}^{(s)}(X \restriction {mL})) &= s.k.L + \sum_{q\in Q}  \sum_{x \in  \{0,1\}^L}  {N(q,x)} \; \log( \beta(q,x)) \\
		&= s.k.L + k \sum_{q \in Q} \frac{N(q)}{k}  \sum_{x \in  \{0,1\}^L}   \frac{N(q,x)} {N(q)} \; \log( \beta(q,x)) \\
		&= k \left( s.L +  \sum_{q \in Q} P(q) \sum_{x \in  \{0,1\}^L}   P(x|q) \; \log( \beta(q,x))\right).
	\end{align*}
	
	$d_{G_L}^{(s)}$ succeeds on $X$ only if for infinitely many $k \in \mathbb{N}$,
	
	\begin{align*}
		-  \sum_{q \in Q} P(q)  \sum_{x \in  \{0,1\}^L} P(x|q) \; \log( \beta(q,x)) < s.L.
	\end{align*}
	
	Using the properties of KL divergence \cite{CoverThomas1991}, we can argue that
	
	\begin{align*}
		- \sum_{q \in Q} P(q) \sum_{x \in  \{0,1\}^L} P(x|q) \; \log( P(x|q)) < sL.
	\end{align*}

	Using the following result from information theory  \cite{CoverThomas1991}, $H(X) \leq H(X|S) + H(S)$
	\begin{align*}
		\sum_{x \in  \{0,1\}^L} P(x) \; \log( P(x)) &< sL + H(Q). 
	\end{align*}
	where $H(Q)$ is the entropy of occurrence of states which is atmost $\log(n)$, where $n = |Q|$, the number of states in the original gambler $G$.

	Dividing by L, we have that for infinitely many $k \in \mathbb{N}$, 	
	\begin{align*}
		H_L(X \restriction {kL})\ < \;   s + \frac{\log(n)}{L}.
	\end{align*}
	
	Taking $\liminf$ over k,
	\begin{equation}\label{eq:HLandL}
		\liminf_{k \rightarrow \infty} H_L(X \restriction {kL}) \leq \;   s + \frac{\log(n)}{L}. 
	\end{equation}
	
	Taking infimum over all $L$,
	\begin{align*}
		H(X) = \inf_L \liminf_{k \rightarrow \infty} H_L(X \restriction {kL}) \leq s.
	\end{align*}
	
\end{proof}

We have the theorem by Bourke, Hitchcock and Vinodchandran \cite{bourke2005entropy} showing equivalence between finite state dimension and disjoint block entropy rates. It follows from Lemma \ref{lem3.3} and \ref{lem3.6}.

\begin{theorem} [Bourke, Hitchcock and Vinodchandran \cite{bourke2005entropy}] \label{thm:DisjEquiv}
	For any sequence $X \in \Sigma^\infty$
	\[\dimfs(X) = H^{disjoint}(X).\]
\end{theorem}

\section{Product Gales}

We define the notion of product-gales which is the product of a finite number of $s$-gales.

\begin{definition}[Product-gale]
	Given $s \in [0, \infty]$ and $k \in \N$, the function $d : \Sigma^* \to [0, \infty)$ is called a $k$-product $s$-gale if
	for $0 < i \leq k$, there are $s$-gales $d_i$ such that
	\[d(w) = d_1(w) . d_2(w) ... d_k(w).\]
\end{definition}

\begin{definition} [Success set of product-gale]
	Let $d$ be a $k$-product $s$-gale.
	\begin{enumerate}
		\item $d$ is said to succeed on a sequence $X \in \Sigma^\infty$ if:
		
		\[\limsup_{n \rightarrow \infty} d(X \restriction n) = \infty.\]
		\item The success set of $d$ is
		
		\[S^\infty[d] = \{X \in \Sigma^\infty : d \text{ succeeds on } X\}.\]
	\end{enumerate}
\end{definition}

\subsection{Product gale characterisation of Hausdorff dimension}

Lutz \cite{Lutz03} showed an $s$-gale characterisation of Hausdorff dimension over the Cantor space $\Sigma^\infty$. We show a similar product $s$-gale characterisation.

Lemma \ref{lem:prod-galecondition} shows that the $k^{th}$ square root function of a $k$-product $s$-gale is an $s$- supergale. We use this to show a generalisation of the Kolmogorov inequality in Lemma \ref{lem:KraftIneqProdGale}.
 
\begin{lemma}
	\label{lem:prod-galecondition} 
	Let $d$ be a $k$-product $s$-gale, then for any $w  \in \{0,1\}^*$,
	\[ \sqrt[k]{d(w0)} +  \sqrt[k]{d(w1)} \leq 2^s  \sqrt[k]{d(w)}.\]
\end{lemma}

\begin{proof}
	Let $d = d_1. d_2\cdots d_k$, where each $d_i$ is a constituent $s$-gale of the product gale $d$. 
	
	Take $\beta_i(w) = 2^{-s} \times \frac{d_i(w1)}{d_i(w)}$, the gale condition ensures that $1 - \beta_i(w) = 2^{-s} \times \frac{d_i(w0)}{d_i(w)}$.
	
	\begin{align*}
		\sqrt[k]{d(w0)} +  \sqrt[k]{d(w1)} &= \sqrt[k]{\prod_{i=1}^k d_i(w0)} + \sqrt[k]{\prod_{i=1}^k d_i(w1)} \\
		&= \sqrt[k]{\prod_{i=1}^k 2^s \times (1-\beta_i(w))  \times d_i(w)} + \sqrt[k]{\prod_{i=1}^k 2^s \times \beta_i(w) \times d_i(w)}\\
		&= 2^s \times \sqrt[k]{\prod_{i=1}^k {d_i(w)}} \times \Bigg\{ \sqrt[k]{\prod_{i=1}^k  \beta_i(w)} + \sqrt[k]{\prod_{i=1}^k 1 -  \beta_i(w)} \Bigg\}\\
		&= 2^s \times \sqrt[k]{d(w)} \times \Bigg\{ \sqrt[k]{\prod_{i=1}^k  \beta_i(w)} + \sqrt[k]{\prod_{i=1}^k 1 -  \beta_i(w)} \Bigg\} \\
		&\leq 2^s \times \sqrt[k]{d(w)} \times \frac{1}{k}\Bigg\{ \sum_{i = 1}^{k} \beta_{i}(w) + \sum_{i = 1}^{k} (1 - \beta_{i} (w))\Bigg\} \\
		&= 2^s \times \sqrt[k]{d(w)}.
	\end{align*}
	
	The second equality follows from the $s$-gale condition (Definition \ref{def:s-gale}).  The fourth follows from the fact that $d(w) = \prod_{i=1}^{k} d_i(w)$. The fifth inequality follows using the Arithmetic mean $\geq$ Geometric mean inequality.
\end{proof}

\begin{lemma}
	\label{lem:KraftIneqProdGale}
	
	Let $d$ be a $k$-product $s$-gale and $B \subseteq$ $\{0,1\}^*$ be a prefix set, then for all $w \in \{0,1\}^*$,
	
	\[\sum_{u \in B} 2^{-s|u|} \sqrt[k]{d(wu)} \leq \sqrt[k]{d(w)}.\]
\end{lemma}

\begin{proof}
	
	Using induction, first we show that for all $n \in \mathbb{N}$, the lemma holds for all prefix sets $B \subseteq \{0,1\}^{\leq n}$. That is for all $n \in \N$, $B \subseteq \{0,1\}^{\leq n}$, and $w \in \Sigma^*$,
	
	\[\sum_{u \in B} 2^{-s|u|} \sqrt[k]{d(wu)} \leq \sqrt[k]{d(w)}.\]
	
	This is trivial for $n = 0$. Assume that the lemma holds for $n$, and let $A \subseteq \{0,1\}^{\leq n+1}$ be a prefix set. Let $A_{\leq n} =  \{0,1\}^{\leq n} \cap A$ and $A_{= n} =  \{0,1\}^{= n} \cap A$. 
	Let $A' = \{u \in \{0,1\}^{n} \mid u0 \in A \text{ or } u1 \in A \}$.
	
	\medskip
	
	For all $w \in \{0,1\}^*$,
	
	\begin{align*}
		\sum_{u \in A_{=n+1}} 2^{-s|u|}  \sqrt[k]{d(wu)} &= 2^{-s(n+1)} \sum_{u \in A_{=n+1}} \sqrt[k]{d(wu)} \\
		& \leq 2^{-s(n+1)} \sum_{v \in A'} \sqrt[k]{d(wv0)} + \sqrt[k]{d(wv1)}.
	\end{align*}
	
	Applying  Lemma \ref{lem:prod-galecondition},
	
	\begin{align*}
		\sum_{u \in A_{=n+1}} 2^{-s|u|}  \sqrt[k]{d(wu)}	& \leq 2^{-s(n+1)} \sum_{v \in A'} 2^s \sqrt[k]{d(wv)} \\
		& = \sum_{v \in A'} 2^{-s|u|} \sqrt[k]{d(wv)}.
	\end{align*}

	We have that,
	\begin{align*}
		\sum_{u \in A} 2^{-s|u|}  \sqrt[k]{d(wu)} & = \sum_{u \in A_{\leq n}} 2^{-s|u|}  \sqrt[k]{d(wu)} + \sum_{u \in A_{=n+1}} 2^{-s|u|}  \sqrt[k]{d(wu)}\\
		&\leq  \sum_{u \in A_{\leq n}} 2^{-s|u|}  \sqrt[k]{d(wu)} +  \sum_{u \in A'} 2^{-s|u|}  \sqrt[k]{d(wu)}.
	\end{align*}
	
	Let $B = A_{\leq n} \cup A'$.	
	Note that $B$ is a prefix set and $A_{\leq{n}} \cap A' = \emptyset$. Then,
	
	\begin{align*}
		\sum_{u \in A} 2^{-s|u|}  \sqrt[k]{d(wu)} &\leq \sum_{u \in B} 2^{-s|u|} \sqrt[k]{d(wu)}.
	\end{align*}
	
	Since $B \subseteq \{0,1\}^{\leq n}$, it follows by induction hypothesis that for all $w \in \{0,1\}^*$,
	
	\begin{align*}
		\sum_{u \in A} 2^{-s|u|}  \sqrt[k]{d(wu)} 
		&\leq \sqrt[k]{d(w)}.
	\end{align*}
	
	Thus for all $n \in \mathbb{N}$, the lemma holds for prefix sets $B \subseteq \{0,1\}^{\leq n}$.
	
	Let $B$ be an arbitrary prefix set, then for all $w \in \{0,1\}^*$,
	
	\begin{align*}
		\sum_{u \in B} 2^{-s|u|}  \sqrt[k]{d(wu)} &= \lim_{n \to \infty} \sum_{u \in B \cap \{0,1\}^n} 2^{-s|u|}  \sqrt[k]{d(wu)} \\
		& \leq \sqrt[k]{d(w)}.
	\end{align*}
	
\end{proof}


\begin{definition} For $\F \subseteq \Sigma^\infty$,
	Let $\mathcal{G}_{prod}(\F)$ denote the set of all $s \in[0, \infty)$ such that there is a $k$ - product $s$-gale $d$, for some $k \in \mathbb{N}$, for which $\F \subseteq S^\infty[d]$.
\end{definition}

We show a product gale characterisation of Hausdorff dimension over $\Sigma^\infty$. The proof is a generalisation of the proof by Lutz \cite{Lutz2003}.

\begin{theorem} \label{thm:ProdGaleCharnDim}
	For all $\F \subseteq \Sigma^\infty$, \[\dim(\F) = \inf \mathcal{G}_{prod}(\F).\]
\end{theorem}

\begin{proof}
	
	It suffices to show that for all $s \in [0, \infty)$, 
	
	\[H^s(\F) = 0 \iff s \in \mathcal{G}_{prod}(\F).\]
	
	Lutz (Theorem 3.10 in  \cite{Lutz03}) showed that if $H^s(\F) = 0$, then there exists an $s$-gale $d$ that succeeds on $\F$. Since an $s$-gale is also a product $s$-gale, it follows that $H^s(\F) = 0 \implies s \in \mathcal{G}_{prod}(\F)$.
	
	Conversely, assume that a $k$-product $s$-gale $d$ succeeds on $\F$. 
	It suffices to show that $H^s(X) \leq 2^{-n}$, for any $n \in \mathbb{N}$.
	
	Let \[a_\ell = 1 + max\{d(w) \mid w \in\{0,1\}^{\leq \ell}\}.\]
	
	And let \[A_\ell = \{w\in \{0,1\}^* \mid d(w) \geq 2^{nk} \cdot a_\ell \text{ and } (\forall v) \; [v \sqsubset w \implies d(v) < 2^{nr} \cdot a_\ell ]\}.\]
	
	By definition, $A_\ell$ is a prefix set. Also, we have $A_\ell \subseteq \Sigma^{\geq \ell}$. 
	It is also clear that $\F \subseteq S^\infty[d] \subseteq \bigcup_{w \in A_k} C_w.$
	
	By Lemma \ref{lem:KraftIneqProdGale}, we have
	
	\[\sqrt[k]{d(\lambda)} \geq \sum_{w\in A_\ell}2^{-s|w|} \sqrt[k]{d(w)}.\]
	
	Since $d(w) \geq 2^{nk} \cdot d(\lambda)$,
	
	\[\sqrt[k]{d(\lambda)} \geq  2^n . \sqrt[k]{d(\lambda)} \sum_{w\in A_\ell}2^{-s|w|}.\]
	
	Taking $d(\lambda) = 1$, it follows that
	
	\[\sum_{w\in A_\ell}2^{-s|w|} \leq 2^{-n}.\]
	
	Therefore, for all $\delta >0$, taking $\ell$ such that $2^{-\ell} < \delta$, we have \[\mathcal{H}^s_\delta(\F) \leq 2^{-n}\].
	
	From this it follows that for all $n \in \N$,
	
	\[\mathcal{H}^s(\F) = \lim_{\delta \to 0} \mathcal{H}^s_\delta(\F)  \leq 2^{-n}.\]
	
\end{proof}

\section{$k$-bet finite state gamblers}

We introduce the notion of a $k$-bet finite state gambler. In this model, at each state, the gambler can place $k$ separate bets on the next bit. 

\begin{definition}[$k$-bet finite-state gambler]
	A $k$-bet finite-state gambler ($k$-bet FSG) is a 5-tuple
	\[
	G = (Q, \delta, \vec{\beta}, q_0, c_0),
	\]
	where
	\begin{itemize}
		\item \( Q \) is a nonempty, finite set of states,
		\item \( \delta : Q \times \{0, 1\} \to Q \) is the transition function,
		\item \( \vec{\beta} : Q \to (\Q \cap [0,1])^k \) is the betting function,
		\item \( q_0 \in Q \) is the initial state, and
		\item \(c_0\), is the initial capital, a nonnegative rational number.
	\end{itemize}
\end{definition}

The normal finite-state gambler is a special case of the  $k$-bet s-gale when $k = 1$.

We define the notion of an $s$-gale induced by a $k$-bet finite state gambler $G$.

\begin{definition} [$s$-product-gale induced by $k$-bet FSG]
	Given $s \in [0,\infty)$, the $s$-product-gale induced by the $k$-bet FSG \( G \) is the function 
	\[
	d_{G}^s : \{0, 1\}^* \to [0, \infty)
	\]
	defined by the recursion:
	\[
	d_{G}^s(\lambda) = c_0,
	\]
	and
	\[
	d_{G}^s(wb) =	d_{G}^s (w) \cdot  \prod_{i=1}^{k}   2^s \cdot  \left[(1 - b) \cdot (1 - \beta_i(w)) + b\cdot\beta_i(w)\right].
	\]
\end{definition}

\textbf{Note:} The notion of multi-bet finite-state $s$-gales is fundamentally different from the notion of multi-account finite-state $s$-gales in \cite{Dai2001}. In the latter notion, $k$- separate  bets are maintained by the automata. The key difference is that the $s$-gale in the multi-account case  is a sum of the $s$-gales $d_i^s$ induced by the $k$  accounts. Therefore, the summation function $d^s$ is still an $s$-gale, as gales are closed under finite addition. 

The key difference in the multi-account case is that the \emph{product} of the $s$-gales induced by the $k$ bets are taken as the overall $s$-product gale. This does not effectively become an $s$-gale as $s$-gales are not closed under finite multiplication. This however becomes a $k$-product $s$-gale.

\begin{definition}
	Let $\mathcal{G}_{FS}^{\mbet}(X)$ denote the set of all $s \in[0,1]$ such that there exists a $k \in \N$ and a $k$-bet finite state $s$-gale $d^s_G$ for which $X \in S^\infty[d_G^s]$.
\end{definition}

Finite state dimension is defined using single-bet finite state $s$-gales. We define an analogous notion for multi-bet finite state $s$-product-gales. We later show (Theorem \ref{thm:MbetFsdEqvl}) that both these notions are equivalent.
   
\begin{definition}
	The \emph{multi-bet finite state dimension} of a sequence $X \in \Sigma^\infty$ is
	\[\dimfs^{\mbet}(X) = \inf \mathcal{G}_{FS}^{\mbet}(X).\]
\end{definition}

\section{Finite state dimension and sliding block entropy rates}

In this section, we show the equivalence between multi-bet finite state dimension and sliding block entropy rates.

\subsection{Betting on sliding entropy rates}

For sliding block entropy rates, we have the following Lemma. The proof proceeds in the same lines as proof of Lemma \ref{lem3.1}.

\begin{lemma} 
	\label{lem3.2} 
	Let $X \in \Sigma^\infty$ with $H^{sliding}(X) = s$. 
	For any $ s' >s$, there exists a block length  $\ell \in \mathbb{N}$, and a probability distribution function $\mathbb{P} : \{0,1\}^\ell \rightarrow \Q \cap [0,1]$, such that for infinitely many $k \in \mathbb{N}$,
	
	
	\[   \sum_{w \in \{0,1\}^\ell}  P(w,X \restriction k) \; \log(1 / \mathbb{P}(w))) < \ell \; (s' - d)\]
	
	for some $d < s'$.
\end{lemma}

We will use $\mathbb{P}(w)$'s obtained from Lemma \ref{lem3.2} to construct an $\ell$-bet $s$-gale that bets on strings on length $\ell$ and succeeds on strings $X$ for which $H^{sliding}_\ell(X) < s$. 

The key idea is that $\ell$-bet gales can place bets on a $\ell$-length sliding block. The states keep track of the current $\ell$-length block.
 In the diagram below, we illustrate (just) the $\ell$-bets placed and the transitions when  $\ell = 3$ and the next input is $a_1.a_2.a_3$. The bets are placed according to the distribution $p$ on $\Sigma^\ell$.
 Here $p_1 = p(a_1) = \sum_{w \in \Sigma^2}p(a_1.w), p_2 = p(a_1.a_2|a_1) , p_3 = p(a_1a_2a_3 | a_1 a_2)$. Note that only the transitions on input $a_1.a_2.a_3$ is shown.

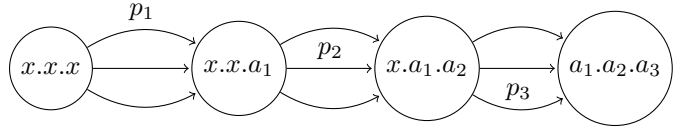
\begin{center}
\begin{tikzpicture} [shorten >=1pt, node distance=2.5cm, on grid, auto]
	
	\node[state] (q0) {$x.x.x$};
	\node[state, right of=q0] (q1) {$x.x.a_1$};
	\node[state, right of=q1] (q2) {$x.a_1.a_2$};
	\node[state, right of=q2] (q3) {$a_1.a_2.a_3$};
	
	\path[->] (q0) edge node {} (q1);
	\path[->] (q0) edge [bend left] node {$p_1$} (q1);
	\path[->] (q0) edge [bend right] node {} (q1);
	
	\path[->] (q1) edge node {$p_2$} (q2);
	\path[->] (q1) edge [bend left] node { } (q2);
	\path[->] (q1) edge [bend right] node { } (q2);
	
	\path[->] (q2) edge node {} (q3);
	\path[->] (q2) edge [bend left] node {} (q3);
	\path[->] (q2) edge [bend right] node {$p_3$} (q3);
	
\end{tikzpicture}
\end{center}

We formalise this idea in the following lemma.

\begin{lemma} For all $X \in \Sigma^\infty$,
	
	\[\dimfs^{\mbet}(X) \leq H^{sliding}(X) \]
\end{lemma}

\begin{proof}
	For any $X \in \Sigma^\infty$, and $ s > H^{sliding}(X) $, let $\ell$ be the block length and $\mathbb{P}: \Sigma^\ell \to \Q \cap [0,1]$ be the probability distribution on $\Sigma^\ell$ obtained from Lemma \ref{lem3.2}. 	We design a  FSG $G_{\mathbb{P}}$, so that at each sliding block occurrence of $w \in \Sigma^\ell$, a cumulative bet of $\mathbb{P}(w)$ is placed. 
	
	Consider the following $\ell$-bet finite state gambler,
	$G_{\mathbb{P}} = (Q,\delta,\vec{\beta},q_0)$, whose components are :
	
	\begin{itemize}
		\item $Q = \{0,1\}^{< \ell} \cup \{0,1\}^\ell$.
		\item If $|w| < \ell$, 
		\[\delta(w, b) = wb.\]
		
		For all $w \in \Sigma^ {\ell - 1}$ and $b,b' \in \Sigma$,
		
		\[
		\delta(bw,b') = wb'.
		\]
		
		\item If $|w| < \ell$. Then, for $1 \leq i \leq \ell$ , $$\beta_i(w) = 1/2.$$
		
		For any $w \in \Sigma^\ell$, let $w = a_0 a_1 \dots a_{\ell-1}$,
		for $1 \leq i \leq \ell$, 
		
		\[\beta_{i} (w) = \mathbb{P}(a_i \dots a_{\ell -1}.1) / \mathbb{P}(a_i \dots a_{\ell -1}). \]
		
		where for $v \in \Sigma^{< \ell}$, $\mathbb{P}(v) = \sum\limits_{w \in \Sigma^\ell : v \prefix w} \mathbb{P}(w)$.

		\item $q_0 = \lambda$.
		
	\end{itemize}

	\medskip
	
	Consider the $s$-gale $d_{G_{\mathbb{P}}}^{s'}$ induced by $G_{\mathbb{P}}$. If the sliding block $w_i \in \Sigma^\ell$ has appeared $n_i$ times in $X \restriction n$,
	\[d_{G_{\mathbb{P}}}^{s}(X \restriction n) \geq 2^c. \prod_{w_i \in  \Sigma^\ell} (2^{s.\ell} \cdot \mathbb{P}(w_i)) ^{n_i}\]
	
	where $c$ is a constant that can encapsulate the losses by the tail bets and first $\ell$ bets placed by $G_\mathbb{P}$ on $X \restriction n$. Let $k = n - \ell + 1$.
	
	Proceeding with the same analysis as proof of Lemma \ref{lem3.3}, we get
	\begin{align*}
		\log(d_{G_{\mathbb{P}}}^{s}(X \restriction n))  &\geq c \sum_{w_i \in  \{0,1\}^l} {n_i} \; (s\cdot \ell  + \log (\mathbb{P}(w_i)) ) \\
		&= c \cdot k  \sum_{w_i \in  \{0,1\}^l} \frac{n_i}{k} \; (s\cdot \ell  + \log (\mathbb{P}(w_i)) )\\ 
		&= c\cdot k  \sum_{w_i \in  \{0,1\}^l} P(w_i,X \restriction k) \; (s\ell + \log(\mathbb{P}(w_i)))\\
		&=  c\cdot k  \left(s\ell - \sum_{w_i \in  \{0,1\}^l} P(w_i,X\restriction k) \; \log(1/{\mathbb{P}(w_i)})\right).
	\end{align*}
	
	From Lemma \ref{lem3.2}, we have for some $d < s$, for infinitely many $k \in \N$,
	\begin{align*}
		\log(d_{G_{\mathbb{P}}}^{s}(X \restriction k\ell)) \;\; 
		&\geq  k \cdot c \cdot(s\ell - \ell(s-d) )\\
		&\geq  k\cdot c\cdot d \cdot \ell.
	\end{align*}
	Since $k$ is unbounded, we see that the $\ell$-bet $s$-gale induced by $G_{\mathbb{P}}$ succeeds on $X$.
	Hence, for any $s >H^{sliding}(X)$, $\dimfs^{\mbet}(X) \leq s$. 
	
	From this, it follows that for all $X \in \Sigma^\infty$, $\dimfs^{\mbet}(X) \leq H^{sliding}(X)$.
\end{proof}

\subsection{Sliding entropy rates from gale win}

We now show that the sliding block entropy rates of a sequence $X$ is less than or equal to $s$ if a $\ell$-product $s$-gale corresponding to $\ell$-bet FSG $G$ suceeds on it. The idea is that given a multiple $L = \ell m$ of $\ell$, we construct an $L$-bet gambler $G'$ that places $m$ copies of the bets placed by $G$. We analyse $G'$ to show that $H^{sliding}_L(X) \leq s + c/L$ for a constant $c$.

\begin{lemma}
	\label{lem3.5}
	
	For every $X \in \Sigma^\infty$,
	
	\[ H^{sliding}(X)  \leq \dimfs^{\mbet}(X). \]
\end{lemma}

\begin{proof}
	
	Let a finite state s-gale $d_{G}^{(s)}$ corresponding to a $\ell$-bet gambler $G = (Q,\delta,\beta,q_0)$ succeed on a sequence $X \in \Sigma^\infty$. Given $L = \ell m$, for $m \in \N$, consider the $L$-bet finite state gambler,
	$G'= (Q',\delta',\vec{\beta},q_0')$, whose components are :
	
	\begin{itemize}
		\item $Q' = Q$.
		\item For all $q \in Q'$  and $b \in \Sigma$,
		\[\delta'(q,b) = \delta(q,b).\]
		\item For all $q \in Q'$  and $1 \leq i \leq L$,
		\[\beta_{i} (q) = \beta_j(q). \]
		where $j = i/\ell$.
		\item $q_0' = q_0.$
	\end{itemize}
	
	It is easy to see that if the s-gale $d_{G}^{(s)}$ succeeds on $X$, then the $s$-gale $d_{G'}^{(s)}(x) = $ succeeds on $X$ as we have that $d_{G'}^{(s)}(x) = (d_{G}^{(s)}(x))^m$.
	
	Now for any $n \in \mathbb{N}$, 
		\begin{align*}
		d_{G'}^{(s)}(X \restriction {n}) =  \; \prod_{q \in Q} \prod_{x \in  \{0,1\}^L}  \{ {2^{s.L} \cdot \beta(q,x)} \} ^{N(q,x)}
	\end{align*}
	where $N(q,x)$ is the number of times $x$ occurs as the next $L$ symbols in $X \restriction n$ when $G'$ is at state $q$ and $$\beta(q,x) = \prod_{i=0}^{L-1} x[i] \cdot \beta_{i+1}(\delta^*(q,x[0:i-1])) + (1- x[i]) \cdot (1 - \beta_{i+1}(\delta^*(q,x[0:i-1])))$$ is the cumulative bet placed by $G'$ on the sliding block $x \in \Sigma^L$ at state $q$ in that scenario. 
	
	Taking logarithm,
	\begin{align*}
		\log (d_{G'}^{(s)}(X \restriction {n})) &= s\cdot L \cdot n +  \sum_{q\in Q}  \sum_{x \in  \{0,1\}^L} {N(q,x)} \; \log( \beta_i(q,x)) \\
		&= s\cdot L \cdot n + n  \sum_{q \in Q} \frac{N(q)}{n}  \sum_{x \in  \{0,1\}^L} \frac{N(q,x)} {N(q)} \; \log( \beta_i(q,x)) \\
		&= n \left( s \cdot L +  \sum_{q \in Q} P(q) \sum_{x \in  \{0,1\}^L} P(x|q) \; \log( \beta_i(q,x))\right).
	\end{align*}
	
	$d_{G_L}^{(s)}$ succeeds on $X$ only if for infinitely many $n \in \mathbb{N}$,
	
	\begin{align*}
		-  \sum_{q \in Q} P(q)  \sum_{x \in  \{0,1\}^L} P(x|q) \; \log( \beta(q,x)) < sL.
	\end{align*}
	
	Using the properties of KL divergence \cite{CoverThomas1991}, we can argue that
	
	\begin{align*}
		- \sum_{q \in Q} P(q) \sum_{x \in  \{0,1\}^L} P(x|q) \; \log( P(x|q)) < sL.
	\end{align*}

	Using the following result from information theory  \cite{CoverThomas1991}, $H(X) \leq H(X|S) + H(S)$, we have that for infinitely many $n \in \N$,
	\begin{align*}
		\sum_{x \in  \{0,1\}^L} P(x, X \restriction n) \; \log( P(x, X \restriction n)) &< sL + H(Q). 
	\end{align*}
	where $H(Q)$ is the entropy of occurrence of states which is atmost $\log(|Q|)$, where $ |Q|$ is the number of states in $G'$.

	Dividing by $L$, we have that for infinitely many $n \in \mathbb{N}$, 	
	\begin{align*}
		H_L^{sliding}(X \restriction {n})\ < \;   s + \frac{\log(|Q|)}{L}.
	\end{align*}
	
	Taking $\liminf$ over $n$,
	\begin{equation}\label{eq:HLandL}
		\liminf_{n \rightarrow \infty} H_L^{sliding}(X \restriction {k}) \leq \;   s + \frac{\log(|Q|)}{L}. 
	\end{equation}
	
	Taking infimum over all $L$,
	\begin{align*}
		H^{sliding}(X) = \inf_L \liminf_{k \rightarrow \infty} H_L^{sliding}(X \restriction {k}) \leq s.
	\end{align*}
	
\end{proof}

\subsection{Entropy rates and dimension}

 We have the following theorem that shows equivalence between multi-bet finite state dimension and sliding block entropy rates.
 
\begin{theorem} \label{thm:slidingeqvl}
	For all $X \in \Sigma^\infty$,
	\[\dimfs^{\mbet}(X) = H^{sliding}(X).\]
\end{theorem}

\section{Equivalence between sliding and disjoint entropy rates}

\subsection{Disjoint entropy rates from $k$-bet gale win}
We now show that the disjoint block entropy rates of a sequence $X$ is less than or equal to $s$ if a $\ell$-product $s$-gale corresponding to $\ell$-bet FSG $G$ suceeds on it. The idea is that given a length $L \in \N$, we  stretch $G$ into $L$ copies and construct a extended $k$-bet gambler $G_L$. We analyse $G_L$ to show that $H^{sliding}_L(X) \leq s + c/L$ for a constant $c$.

\begin{lemma}
	\label{lem:DisjointEntropyMbetgales}
	
	For every $X \in \Sigma^\infty$,
	
	\[ H^{disjoint}(X)  \leq \dimfs^{\mbet}(X). \]
\end{lemma}

\begin{proof}
	
	Let a finite state s-gale $d_{G}^{(s)}$ corresponding to a $k$-bet gambler $G = (Q,\delta,\beta,q_0)$ succeed on a sequence $X \in \Sigma^\infty$. Given $L \in \N$, consider the extended $k$-bet finite state gambler,
	$G_L= (Q',\delta',\beta',q_0')$, whose components are
	
	\begin{itemize}
		\item $Q' = Q \times [L] $
		\item For all $(q,n) \in Q'$  and $b \in \Sigma$,
		\[
		\delta'((q,n),b) = 
		\begin{cases}
			(\delta(q,b),n+1)& \text {if } n < L-1\\
			(\delta(q,b),0) & \text {if } n = L-1\\
		\end{cases}
		\]
		\item For all $(q,n) \in Q'$  and $b \in \Sigma$ and $i \leq k$,
		\[
		\beta'_i((q,n)) = \beta_i(q).
		\]
		\item $q_0' = (q_0,0)$
	\end{itemize}
	
It is easy to see that if the s-gale $d_{G}^{(s)}$ succeeds on $X$, then the $s$-gale $d_{G_L}^{(s)}$ succeeds on $X$. 
	
	Now for any $m \in \mathbb{N}$, 
	\begin{align*}
		d_{G_L}^{(s)}(X \restriction {mL}) =  \prod_{q \in Q} \prod_{x \in  \{0,1\}^L} \prod_{i = 1}^k \{ {2^{s.L} \times \beta_i(q,x)} \} ^{N(q,x)}
	\end{align*}
	where $N(q,x)$ is the number of times $x$ has occurred as the input when $G_L$ was at state $(q,0)$. 	 
	$$\beta_i(q,x) = \prod_{j=0}^{L-1} x[j] \cdot \beta_{i+1}(\delta^*(q,x[0:i-1])) + (1 - x[j]) \cdot (1 - \beta_{i+1}(\delta^*(q,x[0:i-1]))) $$ is the $i^{th}$ cumulative bet placed by $G_L$ on $x$ in that scenario.
	
	Taking logarithm,
	\begin{align*}
		\log (d_{G_L}^{(s)}(X \restriction {mL})) &= s.k.m.L + \sum_{q\in Q}  \sum_{x \in  \{0,1\}^L} \sum_{i \in  1}^k {N(q,x)} \; \log( \beta_i(q,x)) \\
		&= s.k.m.L + m \sum_{q \in Q} \frac{N(q)}{m}  \sum_{x \in  \{0,1\}^L}  \sum_{i \in  1}^k \frac{N(q,x)} {N(q)} \; \log( \beta_i(q,x)) \\
		&= m \left( s.k.L +  \sum_{q \in Q} P(q) \sum_{x \in  \{0,1\}^L}  \sum_{i \in  1}^k P(x|q) \; \log( \beta_i(q,x))\right).
	\end{align*}
	$d_{G_L}^{(s)}$ succeeds on $X$ only if for infinitely many $m \in \mathbb{N}$,
	\begin{align*}
		-  \sum_{i \in  1}^k \sum_{q \in Q} P(q)  \sum_{x \in  \{0,1\}^L} P(x|q) \; \log( \beta_i(q,x)) < s.k.L.
	\end{align*}
	Using the properties of KL divergence \cite{CoverThomas1991}, we can argue that
	\begin{align*}
		- \sum_{q \in Q} P(q) \sum_{x \in  \{0,1\}^L} P(x|q) \; \log( P(x|q)) < sL.
	\end{align*}
	Using the following result from information theory  \cite{CoverThomas1991}, $H(X) \leq H(X|S) + H(S)$
	
	\begin{align*}
	\sum_{x \in  \{0,1\}^L} P(x) \; \log( P(x)) &< sL + H(Q). 
	\end{align*}
	where $H(Q)$ is the entropy of occurrence of states which is atmost log(n), where $n = |Q|$, the number of states in the original gambler $G$.

	Dividing by L, we have that for infinitely many $m \in \mathbb{N}$, 	
	\begin{align*}
		H_L(X \restriction {mL})\ < \;   s + \frac{\log(n)}{L}.
	\end{align*}
	
	Taking $\liminf$ over k,
	\begin{equation}\label{eq:HLandL}
		 \liminf_{k \rightarrow \infty} H_L(X \restriction {kL}) \leq \;   s + \frac{\log(n)}{L}. 
	\end{equation}
	
	Taking infimum over all $L$,
	\begin{align*}
		H^{disjoint}(X) = \inf_L \liminf_{k \rightarrow \infty} H_L(X \restriction {kL}) \leq s.
	\end{align*}
\end{proof}

Therefore we have the following theorem.

\begin{theorem} \label{thm:mbetDisjeqvl}
	For all $X \in \Sigma^\infty$,
	\[H^{disjoint} (X) = \dimfs^\mbet(X).\]
\end{theorem}
\begin{proof}
	From Lemma \ref{lem:DisjointEntropyMbetgales}, we have that $H^{disjoint} (X) \leq \dimfs^\mbet(X)$. From Lemma \ref{lem3.3}, we have $H^{disjoint} (X) \geq \dimfs(X)$. 
	
	Since by definition, $\dimfs(X) \geq \dimfs^\mbet(X)$, it follows that, \\ $H^{disjoint} (X) \geq \dimfs^\mbet(X)$.
\end{proof}

Therefore, we have a new characterisation of finite-state dimension using $k$-bet gales. 
The following theorem follows from Theorem \ref{thm:DisjEquiv} and Theorem \ref{thm:mbetDisjeqvl}.

\begin{theorem} \label{thm:MbetFsdEqvl}
	For all $X \in \Sigma^\infty$,
	\[\dimfs(X) = \dimfs^{\mbet}(X).\]
\end{theorem}

Kozachinskiy, and Shen \cite{KozachinskiyShen2021} showed that the sliding and disjoint formulations of block entropy rates are equivalent. Note that they use the terminology aligned(disjoint) and non-aligned(sliding) entropy. This result also from Theorem \ref{thm:DisjEquiv}, Theorem \ref{thm:slidingeqvl} and Theorem \ref{thm:mbetDisjeqvl}.

\begin{theorem} [Kozachinskiy, and Shen \cite{KozachinskiyShen2021}]
	For all $X \in \Sigma^\infty$,
	\[H^{sliding}(X) = H^{disjoint}(X).\]
\end{theorem}

Summarising, we now have the following equivalences.

\begin{theorem} 
	For all $X \in \Sigma^\infty$,
	\[H^{sliding}(X) = H^{disjoint}(X) = \dimfs(X) = \dimfs^{\mbet}(X).\]
\end{theorem}

\section*{Acknowledgments}

The author thanks Satyadev Nandakumar and Subin Pulari for their insightful comments and helpful discussions. 
\bibliography{main}

\begin{thebibliography}{10}

\bibitem{bourke2005entropy}
Chris Bourke, John~M. Hitchcock, and N.~V. Vinodchandran.
\newblock Entropy rates and finite-state dimension.
\newblock {\em Theoret. Comput. Sci.}, 349(3):392--406, 2005.

\bibitem{CoverThomas1991}
Thomas~M. Cover and Joy~A. Thomas.
\newblock {\em Elements of information theory}.
\newblock Wiley Series in Telecommunications. John Wiley \& Sons, Inc., New
  York, 1991.
\newblock A Wiley-Interscience Publication.

\bibitem{Dai2001}
Jack~J. Dai, James~I. Lathrop, Jack~H. Lutz, and Elvira Mayordomo.
\newblock Finite-state dimension.
\newblock In {\em Automata, languages and programming}, volume 2076 of {\em
  Lecture Notes in Comput. Sci.}, pages 1028--1039. Springer, Berlin, 2001.

\bibitem{Falc03}
Kenneth Falconer.
\newblock {\em Fractal geometry}.
\newblock John Wiley \& Sons, Inc., Hoboken, NJ, second edition, 2003.
\newblock Mathematical foundations and applications.

\bibitem{Hausdorff1919}
F.~Hausdorff.
\newblock Dimension und {\"a}usseres {M}ass.
\newblock {\em Mathematische Annalen}, 79:157--179, 1919.

\bibitem{KozachinskiyShen2019}
Alexander Kozachinskiy and Alexander Shen.
\newblock Two characterizations of finite-state dimension.
\newblock In {\em Fundamentals of computation theory}, volume 11651 of {\em
  Lecture Notes in Comput. Sci.}, pages 80--94. Springer, Cham, 2019.

\bibitem{KozachinskiyShen2021}
Alexander Kozachinskiy and Alexander Shen.
\newblock Automatic {K}olmogorov complexity, normality, and finite-state
  dimension revisited.
\newblock {\em J. Comput. System Sci.}, 118:75--107, 2021.

\bibitem{Lutz03}
Jack~H. Lutz.
\newblock Dimension in complexity classes.
\newblock {\em SIAM J. Comput.}, 32(5):1236--1259, 2003.

\bibitem{Lutz2003}
Jack~H. Lutz.
\newblock The dimensions of individual strings and sequences.
\newblock {\em Inform. and Comput.}, 187(1):49--79, 2003.

\bibitem{fsrdconf}
Satyadev Nandakumar, Subin Pulari, and Akhil S.
\newblock Finite-state relative dimension, dimensions of ap subsequences and a
  finite-state van lambalgen's theorem.
\newblock In Ding-Zhu Du, Donglei Du, Chenchen Wu, and Dachuan Xu, editors,
  {\em Theory and Applications of Models of Computation}, pages 334--345, Cham,
  2022. Springer International Publishing.

\bibitem{Pillai1940}
S~S Pillai.
\newblock On normal numbers.
\newblock {\em Proc. Indian Acad. Sci. (Math. Sci.)}, 12(2), August 1940.

\bibitem{SchSti72}
C.~P. Schnorr and H.~Stimm.
\newblock Endliche automaten und zufallsfolgen.
\newblock {\em Acta Informatica}, 1:345--359, 1972.

\bibitem{ZivLem78}
J.~Ziv and A.~Lempel.
\newblock Compression of individual sequences via variable rate coding.
\newblock {\em IEEE Transaction on Information Theory}, 24:530--536, 1978.

\end{thebibliography}
\bibliographystyle{plain}

\end{document}